%% file: FinalVersion.tex
\documentclass[final]{dmtcs-episciences}
\usepackage[utf8]{inputenc}
\usepackage[T1]{fontenc}
\usepackage{amsmath, amsthm, amsfonts, amssymb}
\usepackage{graphicx}
\graphicspath{{./FinalFigures}}
\usepackage{pdfpages}
\usepackage{xspace}
\usepackage{hyperref}
\usepackage[capitalise]{cleveref}
\usepackage{stmaryrd}
\usepackage{bm}
\usepackage{doi}
\usepackage{thm-restate}
\usepackage[round]{natbib}

\newcommand{\I}{\ensuremath{\mathcal{I}}\xspace}
\newcommand{\R}{\ensuremath{\mathbb{R}}\xspace}
\newcommand{\Z}{\ensuremath{\mathbb{Z}}\xspace}

\newcommand{\F}{\ensuremath{\mathbb{F}}\xspace}

\newcommand{\0}{\ensuremath{\mathbf{0}}\xspace}
\newcommand{\1}{\ensuremath{\mathbf{1}}\xspace}
\newcommand{\x}{\ensuremath{\mathbf{x}}\xspace}
\newcommand{\y}{\ensuremath{\mathbf{y}}\xspace}
\renewcommand{\a}{\ensuremath{\mathbf{a}}\xspace}
\renewcommand{\b}{\ensuremath{\mathbf{b}}\xspace}
\renewcommand{\c}{\ensuremath{\mathbf{c}}\xspace}
\renewcommand{\d}{\ensuremath{\mathbf{d}}\xspace}

\newcommand{\II}{\ensuremath{\mathbf{\boldsymbol{\infty}}}\xspace}

\newcommand{\NP}{\ensuremath{\textrm{NP}}\xspace}
\newcommand{\PSPACE}{\ensuremath{\textrm{PSPACE}}\xspace}
\newcommand{\ER}{\texorpdfstring{\ensuremath{\exists\R}}{ER}\xspace}
\newcommand{\problemname}[1]{\textnormal{\textsc{#1}}\xspace}
\newcommand{\ETR}{\problemname{ETR}}
\newcommand{\ETRAMI}{\problemname{ETRAMI}}
\newcommand{\Distinct}{\problemname{Distinct-ETR}}
\newcommand{\Strict}{\problemname{STRICT-INEQ}}
\newcommand{\Feasibility}{\problemname{Feasibility}}

\newcommand{\matroidrealizability}{\problemname{Matroid \R-Representability}}
\newcommand{\stretchability}{\problemname{stretchability}}
\newcommand{\representable}{representable\xspace}
\newcommand{\rank}{{\sf rk}\xspace}
\newcommand{\wordRAM}{\textnormal{word RAM}\xspace}
\newcommand{\realRAM}{\textnormal{real RAM}\xspace}

\newtheorem{theorem}{Theorem}

\newtheorem{lemma}[theorem]{Lemma}

\theoremstyle{definition}

\title{Representing Matroids over the Reals is \ER-complete}
\author{Eun Jung Kim\affiliationmark{1} \and Arnaud de Mesmay\affiliationmark{2} \and Tillmann Miltzow\affiliationmark{3}\thanks{T. M. is generously supported by the Netherlands Organisation for Scientific Research (NWO) under project no. VI.Vidi.213.150. }}

\affiliation{KAIST, Daejeon, South Korea and CNRS, Paris, France \\
Univ Gustave Eiffel, CNRS, LIGM, F-77454 Marne-la-Vallée, France \\
Department of Information and Computing Sciences, Utrecht University, The Netherlands}

\publicationdata{vol. 26:2}{2024}{10}{10.46298/dmtcs.10810}{2023-01-14; 2023-01-14; 2024-01-10}{2024-05-13}

\begin{document}

\maketitle

\begin{abstract} \vspace{-1em} A \emph{matroid} $M$ is an ordered pair $(E,\mathcal{I})$, where $E$ is a finite set called the  \emph{ground set} and a collection $\I\subset 2^{E}$ called the \emph{independent sets} which satisfy the conditions: 
    (i) $\emptyset \in \mathcal{I}$, (ii) $I'\subset I \in \mathcal{I}$ implies $I'\in \mathcal{I}$, and (iii) 
    $I_1,I_2 \in \mathcal{I}$ and $|I_1| < |I_2|$ implies that there is an $e\in I_2$ such that $I_1\cup \{e\} \in \mathcal{I}$.
    The \emph{rank} ${\sf rk}(M)$ of a matroid $M$ is the maximum size of an independent set. 
    We say that a matroid $M=(E,\mathcal{I})$ is \emph{representable} over the reals if there is a map
    $\varphi \colon E \rightarrow \mathbb{R}^{{\sf rk}(M)}$ such that $I\in \mathcal{I}$ if and only if 
    $\varphi(I)$ forms a linearly independent set.
    
    We study the problem of \textsc{Matroid \ensuremath{\mathbb{R}}-Representability} over the reals.
    Given a matroid~$M$, we ask whether there is a set of points in the Euclidean space
    representing~$M$.
    We show that \textsc{Matroid \ensuremath{\mathbb{R}}-Representability} is $\exists \mathbb{R}$-complete, already for matroids of rank~$3$.
    The complexity class $\exists\mathbb{R}$ can be defined as the family of algorithmic problems
    that is polynomial-time equivalent to determining if a multivariate polynomial
    with integer coefficients has a real root.
    
    Our methods are similar to previous methods from the literature.
    Yet, the result itself was never pointed out and there is 
    no proof readily available in the language of computer science.
\end{abstract}

\section{Introduction}
Many articles on matroids assume that the matroid is representable, see for example~\citet*{Lovasz80,cameron2022flag,cameron2017excluded}. 
\nocite{cameron2014kinser,cameron2016excluded, cameron2017polytopal}
Representability either heavily simplifies proofs and definitions or is even essential.
We show that the question of representability over the reals is as difficult as the existential theory of the reals,
that is \ER-complete.
The complexity class \ER can be defined as the family of algorithmic problems
that is polynomial-time equivalent to determining if a multivariate polynomial
(with integer coefficients) has a real root, see Section~\ref{S:overview} for an introduction and overview of this complexity class.

\paragraph{Definitions.}
    Before we give a general definition of a matroid, 
    we introduce \emph{vector matroids}.
    Given a matrix $A$ over a field $\F$, we can define the corresponding 
    vector matroid $M[A]=(E,\I)$ as follows.
    The ground set~$E$ of $M[A]$ is formed by the columns of~$A$ 
    and we say that a subset $I\subset E$ is independent in $M[A]$, i.e., $I\in \I$,
    if the columns are linearly independent over $\F$, that is, the zero vector can be linearly combined by the vectors in $I$ only trivially. A set of elements of $E$ which is not independent is said to be dependent. 
    Note that any set of columns containing a zero column is dependent.
    The independent sets of a vector matroid satisfy three simple properties (see below).
    One way to look at matroids is to see them as abstract set
    systems that have those three properties.
    A \emph{matroid} is an ordered pair $(E,\I)$, where $E$ is a finite set called the  \emph{ground set} and a collection $\I\subset 2^{E}$ called the \emph{independent sets} which satisfy the conditions: 
    \begin{itemize}
        \item[(i)] $\emptyset \in \I$,
        \item[(ii)] $I'\subset I \in \I$ implies $I'\in \I$, 
        \item[(iii)] $I_1,I_2 \in \I$ and $|I_1| < |I_2|$ implies that there is an $e\in I_2$ such that $I_1\cup \{e\} \in \I$.
    \end{itemize}
    The \emph{rank} $\rank(M)$ of a matroid $M$ is the maximum size of an independent set. 
    We say that a matroid $M=(E,\I)$ is \textit{\representable} over \F if there is a matrix $A$ over \F such that $M = M[A]$. 
    Note that all columns of $A$ live in a subspace of dimension at most $\rank(M[A])$.
    Therefore, we can assume without loss of generality that the
    columns of $A$ have dimension $\rank(M[A])$. We refer to~\citet*{oxley2006matroid} for more background on matroids.
    
    Given a matroid $M$, if there exists a matrix $A$ over $\mathbb{R}$ such that $M = M[A]$, we say that $M$ is representable over the reals. 
    The algorithmic problem of \matroidrealizability is to test whether a given matroid is representable over the reals. Since we discuss only the real case in this article, we will sometimes say representable as a shorthand.    
    Note that we also need to specify how the matroid $M$ is given.
    In the literature on matroids, one has often an oracle such that
    one can ask the oracle for each set $I$ whether $I\in \I$.
    We will deviate from this practice, as it might be unclear how
    to describe the oracle.
    Instead, we will just list all sets in $\I$ explicitly.
    This does not blow up the description complexity too much
    in our case, as we will deal mainly with constant rank matroids.

    We refer to Section~\ref{S:overview} for more background, motivation and application of matroids.

\paragraph{Geometric Interpretation.}
    If we have a representation $A \in \R^{3\times n}$ of a matroid $M = M[A]$, we can scale an arbitrary column of $A$ by a nonzero number and it stays a valid representation. 
    This is also the case if we scale by $-1$.
    Furthermore, if we rotate $A$ (i.e., multiply it on the left with an orthogonal transformation) it still stays a valid representation. 
    Thus, in case we have a \representable rank-$3$ matroid over \R, 
    we can assume that there is a representation in which 
    all nonzero vectors have their third coordinate equal to~$1$. 
    In this way, we can consider the columns of~$A$ 
    as a point configuration in the plane with $z=1$.
    The property of three vectors being dependent is then equivalent to the corresponding three points lying on a common line.
    \begin{center}
        \includegraphics{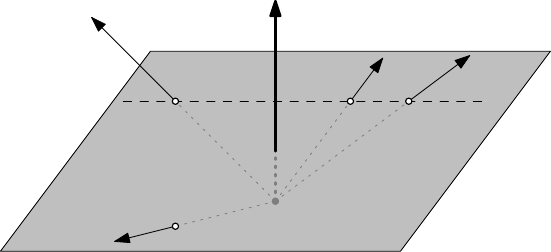}    
    \end{center}

In this geometric interpretation, we could have considered any plane different from the one with $z=1$, as long as it does not contain the origin. This would have yielded a different point configuration. The resulting transformation is called a \emph{projective transformation}. It maps lines to lines, except for one line that disappears. We say that this line is \emph{sent to infinity}. Conversely, for any line in the plane, there exists a projective transformation that sends it to infinity. Throughout this article, we think of representations of rank-$3$ matroids via these point configurations, and thus we will slightly abuse language by calling such a point configuration a representation.

\paragraph{A motivating example.}
One of the standard examples to illustrate realizability is 
the so-called Fano plane.
It is the matroid on seven elements whose maximal independent sets are all the triples except \[\{1,4,7\},\{1,2,3\},\{1,5,6\},\{3,6,7\},\{2,5,7\},\{3,4,5\},\{2,4,6\}.\] 

This can be represented pictorially as in the figure below, where the lines and the circle denote the dependencies.

\begin{center}
\includegraphics[]{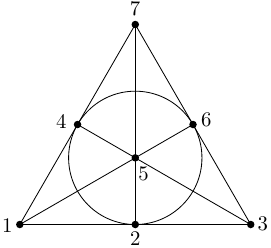}
\end{center}

An immediate question that arises from this picture is whether a picture exists where the circle is not used and the dependencies are all pictured by lines. This is equivalent to asking whether the Fano matroid is representable over the reals. It is well-known not to be~\cite[Proposition~6.4.8]{oxley2006matroid}. 
The problem \matroidrealizability addresses the general question of deciding whether a given matroid can be represented like that, and our main result is that this problem is \ER-complete.

\paragraph{Order Types.}
    We saw above that matroids are an abstraction to describe
    point collinearities in the plane.
    I.e., if we have a rank-3 matroid then every dependent set corresponds to three collinear points.
    Now given a set of points, we are often also interested in the orientation of each triple: either clockwise, counter-clockwise, or collinear.

    \begin{center}
        \includegraphics[]{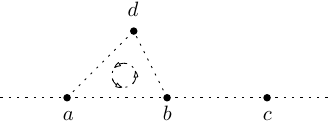}    
    \end{center}
    
    In this example, $\{a,b,c\}$ is collinear and  $(a,b,d)$, $(a,c,d)$, and $(b,c,d)$ are oriented counter-clockwise.
    This leads to the definition of \textit{(abstract) order types},
    which is a pair $O = (E,\chi)$.
    Again, $E$ is a finite set called the ground set. Then \[\chi : \binom{E}{3} \rightarrow \{-1,0,1\}\] is
     a function called a \emph{chirotope} satisfying a few simple properties that are derived from the intuition given above. 
    We say that a point set $P\subset \R^{2}$ \textit{represents} a given order type $O = (E,\chi)$ 
    if $P$ has, for each element $e\in E$, a corresponding point $e'\in P$.
    Furthermore, for each triple $a,b,c\in E$ the corresponding points $a',b',c'\in P$ are oriented according to $\chi(\{a,b,c\})$.
    Note that if we lift every point in $P$ to the plane with $z=1$ as a subset of $\R^{3}$, then we get the following correspondence between $(p_x',p_y',1),(q_x',q_y',1),(r_x',r_y',1)$ and the elements $p,q,r\in E$:
    
    \[
    \text{sign} \, \det \left(\begin{matrix}
    p_x' & q_x' & r_x'\\
    p_y'  & q_y' & r_y' \\
    1  & 1 & 1
    \end{matrix} \right) 
    = \chi(p,q,r).
    \]

    Note that in this specific setup a realization of a rank 3 matroid and order types are closely related.
    While a matroid determines only the collinearities, the order type also determines the orientation of each triple.
    We want to point out that every abstract order type can be represented by a pseudoline arrangement.
    A pseudoline arrangement can be defined as a collection of 
    $x$-monotone curves such that any pair of curves intersects exactly once. The orientation of a triple of pseudolines is defined by the orientation of the triangle that they form (a degenerate triangle corresponding to a zero orientation).
    
        \begin{center}
            \includegraphics[page = 2]{oriented}  
        \end{center}
    
    Note that this pseudoline arrangement corresponds to the order type example given above: $\chi(abc)=0$, $\chi(abd)=1$, $\chi(bcd)=1$ and $\chi(acd)=1$.
    Now, the realizability of the order types is equivalent to the  \stretchability of pseudoline arrangements,
    that is, finding a line arrangement with the same combinatorics as the pseudoline arrangement.
    It is one of the central theorems in the field of the existential theory of the reals that \stretchability is \ER-complete.
    We will use many ideas of that proof for our main result.

    We also want to point out that the notion of an order type can be easily generalized to dimension $d$. 
    The chirotope becomes a function of all $(d+1)$ tuples and tells us the orientation in $d$ dimensions. 
    To illustrate this, if we have four points $a,b,c,e$ in $3$-space, then the points $a,b,c$ lie on a hyperplane $H$. 
    Then the chirotope tells us on which side of $H$ the point $e$ lies, for an orientation of $H$ defined by the three points $a$, $b$ and $c$.

\begin{figure}
\centering
\def\svgwidth{12cm}
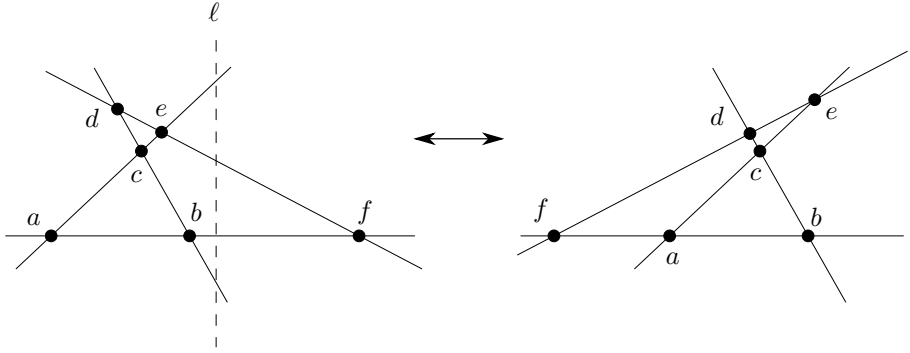
\caption{As $\ell$ separates $f$ from the other points, a projective transformation sending the line $\ell$ to infinity will flip the orientation of the triangles involving $f$ while keeping the other orientations unchanged.}
\label{F:projection}
\end{figure}

    It is important to note that if we take a projective transformation of the plane, we preserve the represented matroid. 
    This is because lines are mapped to lines and points to points.
    However, projective transformations do not preserve the order type of a point set, since a point may end up on the other side of some line, as depicted in Figure~\ref{F:projection}.
    In order to get a closer relationship, we work (sometimes) with matroids endowed with a distinguished line \textit{at infinity} $\ell_\infty$.
    Then we consider \emph{valid} representations of such matroids, which are those where this line is at infinity (i.e., all the points lie on one side of it). This definition extends to rank-$k$ matroids using a hyperplane at infinity and leads to the following definition.

Given an order type $O$, we say that a matroid $M$ \emph{simulates} $O$ if the underlying matroid of $O$ is a subset of $M$ and
if the following conditions are met:
\begin{itemize}
    \item Any representation of the matroid underlying $O$ extends to a representation of $M$.
    \item Any valid representation of $M$ induces a point set representing $O$.
\end{itemize}
    Note that when $M$ simulates $O$, then $M$ has a representation if and only if $O$ has an oriented representation: indeed, starting with a representation of $M$, one can always send the line at infinity to infinity using a projective transformation and thus obtain a valid representation..

\subsection*{Our results}
    Our main theorem is that \matroidrealizability is complete for the existential theory of the reals.

\begin{restatable}{theorem}{selfcontained}
\label{thm:Self-Contained}
\matroidrealizability is \ER-complete.
\end{restatable}

We provide two proofs of \Cref{thm:Self-Contained}. The first one relies on simulating
arbitrary \ETR-formulas using addition and multiplication
and some technical assumptions (see the overview in Section~\ref{S:overview}).
The second proof is somewhat easier, starting from the fact that
order type realizability is \ER-complete, and then simulating order types
using normal matroids.

\begin{restatable}{theorem}{ordertypes}
\label{thm:ordertypes}
Let $k \geq 3$ be a fixed integer. Given a rank-$k$ 
order type
$O$, we can compute in linear time a rank-$k$ matroid $M$ such that $M$ simulates $O$.
\end{restatable}

Theorem~\ref{thm:Self-Contained} easily follows from Theorem~\ref{thm:ordertypes} as deciding whether an order type is representable over the reals is \ER-complete.
This follows from techniques dating back to the proof of the Mn\"{e}v Universality Theorem~\citep{Matousek2014_IntersectionGraphsER,Schaefer2010_GeometryTopology}. 
However, as explained in~\citet*{Matousek2014_IntersectionGraphsER}, the proof that \stretchability is \ER-hard requires a significant number of intricate steps, some of which can be simplified in the setting of \matroidrealizability. 
Indeed, the need for different scales (see for example~\cite[Proof of Theorem~4.6]{Matousek2014_IntersectionGraphsER}), which is the main difficulty in the oriented case, can be completely circumvented in our case.
Therefore, for the sake of completeness and simplicity, we also provide a self-contained proof of Theorem~\ref{thm:Self-Contained}.
We think that it might be educational to first understand the proof of \Cref{thm:Self-Contained}, before one tries to understand the \ER-completeness of \stretchability.

\section{Proof Overview and Background}\label{S:overview}
\subsection*{Proof Overview}
In order to give an idea of the direct proof of \Cref{thm:Self-Contained}, we first sketch an incorrect proof.
Then we point out the issues with this sketch and how we can fix them.

It is folklore that in order to prove \ER-completeness, it is sufficient to find a way to encode variables and some basic operations like addition ($x+y=z$) and multiplication ($xy=z$).
We can force points to lie on a specific line $\ell$ to represent our variables.
Furthermore, using the well-known \textit{von Staudt constructions}, we can simulate all the basic constraints, see \Cref{F:addition-intro} for the construction to simulate addition.
This (almost) describes a rank-3 matroid $M$. 
Furthermore, in any realization of $M$, we can read a valid variable assignment.

\begin{figure}[ht]
    \centering
    \def\svgwidth{9cm}
    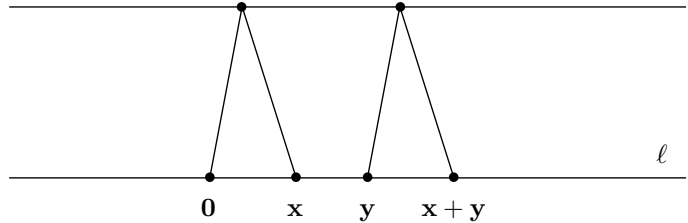
    \caption{Encoding addition geometrically.}
    \label{F:addition-intro}
\end{figure}

The issue with this basic approach is that it could be that there is only one realization of $M$ such that two points coincide, or three points lie on a common line, accidentally. 
If we were able to anticipate this, we could easily specify this in the description of the matroid, but in general, this is not easy.

We circumvent this general position issue with two fixes.
The first fix is to reduce from a version of \ETR where we can assume that all variable values are distinct. We call this variant \Distinct, see \Cref{sec:Distinct}.
The second fix is to observe that when we build the von Staudt construction,  we can ensure that all helper points have enough \textit{freedom} to avoid any coincidences or collinearities with previously defined points, see \Cref{sec:Arithmetic}.
With these two fixes, the above proof sketch works as is explained in the full proof of \Cref{thm:Self-Contained}.

The proof idea of \Cref{thm:ordertypes} goes as follows. Using arithmetic operations, we can give a variable the value $y = x^2\geq 0$. 
Geometrically, this implies that the point representing $y$ is on the same side as the point representing $1$ within the common line $\ell$ (with respect to the point respresenting $0$).
In other words, we can enforce two points to lie on the same side of a line with respect to a given point.
We lift this to half-spaces in the plane and higher dimensions.
In this way, we can enforce consistent orientations of the matroid with the given order type.

\paragraph{Results on \Distinct.}
We define the problem \Distinct as a variant of \ETR as follows (see later for a proper definition of ETR). 
We are given variables $X = \{x_1,\ldots,x_n\}$ and constraints of the form 
\[x+y = z, \quad  x\cdot y = z, \quad x = 1, \quad x>0,\]
for $x,y,z\in X.$
Furthermore, we are promised that there is either no solution at all or there is a solution $(x_1,\ldots,x_n)\in \R^n$ such that $x_i\neq x_j$ for all  $i\neq j$.
We show the following theorem in \Cref{sec:Distinct}, which might be of independent interest.
\begin{theorem}
\label{thm:Distinct-complete}
    \Distinct is \ER-complete.
\end{theorem}

While we could not find any prior proof of Theorem~\ref{thm:Self-Contained} in the literature (hence this work), there are many related works from at least two perspectives. 
First, as already mentioned, when one considers order-types instead of unoriented matroids, Theorem~\ref{thm:Self-Contained} is very well-known.
Second, topological universality theorems have been proved for the real-representability of matroids in the algebraic geometry literature, see for example~\citet*{lafforgue2003chirurgie} and~\citet*{lee2013mnev}. 
While our work uses tools that are similar in spirit to those papers, it differs in that our constructions are arguably simpler and that we specifically focus on proving the computational hardness result, which is not the point of focus of those previous works, and is not entirely equivalent (see discussion below). Furthermore, we believe that it is worthwhile to have a complete proof of Theorem~\ref{thm:Self-Contained} in a purely combinatorial language, as opposed to previous works which employ the language of schemes.

Shor's proof of Mnëv's universality theorem~\cite[Section~4]{Shor1991_Stretchability} introduces an intermediate problem called \textit{the existential theory of totally real ordered variables}, which is very similar to \Distinct but features totally ordered variables $x_1<x_2 \ldots < x_n$ as opposed to just requiring distinctness as in our problem. 
This ordering is desirable when one investigates oriented matroids, and unneeded for unoriented ones. This difference allows for a proof that is arguably simpler than his, or at least different.

\subsection*{Background and Related Work}\label{subsec:related}
\paragraph{Matroids and Greedy}
A practical reason why matroids are relevant for computational purposes is that they capture in a simple way the class of discrete objects where greedy algorithms are successful in finding an optimal solution. 
For example, the standard Kruskal and Prim algorithms to compute a Minimum Spanning Tree in a weighted graph can be abstracted by considering the vector matroid defined by an oriented incidence matrix of the graph (called a \emph{graphic matroid}), and then generalized to compute in polynomial time a maximum or minimum-weight basis for any matroid. 
This property actually characterizes matroids, see for example~\cite[Section~1.8]{oxley2006matroid}

\paragraph{Some applications of representability.}
For an algorithm on matroids, a suitable encoding scheme of the input matroid is needed. A common way is to take the input matroid $M=(E,\I)$ in the form of an independence oracle which answers whether a given subset of the ground elements $E$ is independent or not. There are algorithms which run with a polynomial number of queries to such an oracle. For example, a maximum weight independent set of a given matroid can be computed in this way. However, for many natural matroid properties, it is known that there is no algorithm with polynomially bounded  queries to an independence oracle~\citep{JK1982} including representability over the finite field with two elements $GF(2)$ and connectivity of a matroid. 

A vector representation of a matroid offers a compelling alternative to an
independence oracle as matroid operations can be substantially more efficient 
using matrix operations. 
The Matroid parity problem, a common generalization of graph matching and matroid intersection, is solvable in polynomial time given a vector representation~\citep{Lovasz80} while a super-polynomial number of calls is needed 
under the independence oracle model~\citep{JK1982}. Deciding whether 
the branch-width of a matroid is at most $k$ is a common generalization of 
computing the branch-width, rank-width and carving-width of a graph. 
While there is an algorithm with $n^{O(k)}$ queries on an $n$ element matroid 
for this problem~\citep*{OumS07} under the oracle model, it remains unknown
whether the dependency on $k$ 
in the exponent can be replaced by a uniform constant. 
In contrast, the branch-width of a vector matroid can be computed in 
$f(k)\cdot n^3$ time~\citep*{JeongKO21} when the given representation is over 
a finite field $\F$. 

Another powerful application of a vector representation can be found 
in the theory of 
kernelization in parameterized complexity. A surprising discovery 
of Kratsch and Wahlströhm~\citep{KratschW20} is that for many graph cut problems, compressing the input 
boils down to finding a so-called representative set of a matroid. 
When the said matroid is a vector matroid, a representative set of bounded size 
can be efficiently computed in 
polynomial time~\citep*{Lovasz77,Marx09}. It turns out that solutions to graph cut 
problems can be encoded as independent sets in \emph{gammoids}, which form a well-known 
class of representable matroids and of which a vector representation can be 
constructed in randomized polynomial time.

\paragraph{Oriented Matroids.} One might wonder why we jumped from realizability of matroids to realizability of abstract order types, instead of using the perhaps closer notion of oriented matroids~\citep*{bjorner1999oriented}, for which one can also define a realizability problem and investigate its complexity. The reason is that in our arguments, we reason extensively with point configurations, and the geometric interpretation described above does not adapt directly to oriented matroids, as scaling a column by a negative number could lead to a change of the underlying oriented matroid. The correct framework to connect oriented matroids to point configurations is to only consider \emph{acyclic} oriented matroids~\citep*[Section~1.2.b]{bjorner1999oriented}, that is, those for which the geometric interpretation works readily without a need for rescaling by a negative number. 
This notion of acyclic, oriented matroids coincides with the notion of abstract order types, and so do their realizability problems.
Note that in some of the existing literature, oriented matroids and abstract order types are sometimes described as equivalent.
Therefore, we stress this subtle difference here.

\paragraph{The existential theory of the reals.}
The complexity class \ER (pronounced as `ER', `exists R', or `ETR') has gained a lot of interest in recent years. See the compendium by~\citet*{ERcompendium} for a comprehensive overview.
The complexity class is defined via its canonical complete problem \ETR (short for \emph{Existential Theory of the Reals}. ETR refers to a geometric problem and \ER refers to the complexity class. While there are several different variants of ETR, there is only one complexity class.) and contains all problems that polynomial-time many-to-one reduce to it.
In an \ETR instance, we are given a sentence of the form
\[
    \exists x_1, \ldots, x_n \in \R :
    \varphi(x_1, \ldots, x_n),
\]
where~$\varphi$ is a well-formed and quantifier-free formula consisting of polynomial equations and inequalities in the variables and the logical connectives $\{\land, \lor, \lnot\}$.
The goal is to decide whether this sentence is true.
As an example consider the formula $\varphi(X,Y) :\equiv X^2 + Y^2 \leq 1 \land Y^2 \geq 2X^2 - 1$;
among (infinitely many) other solutions, $\varphi(0,0)$ evaluates to true, witnessing that this is a yes-instance of \ETR.
We
use $|\varphi|$ to denote the \textit{length} of $\varphi$, that is, the number of bits necessary to write down
$\varphi$. 
The solution set of an \ETR-formula is called a semi-algebraic set.
The \textit{(bit)-complexity} of a semi-algebraic set is the shortest length of any formula
defining the set.
It is known that
\[
    \NP \subseteq \ER \subseteq \PSPACE
    \text{.}
\]
Here the first inclusion follows because a \problemname{SAT} instance can trivially be written as an equivalent \ETR instance.
The second inclusion is highly non-trivial and was first proven by Canny in his seminal paper~\citep*{Canny1988_PSPACE}.

Note that the complexity of working with continuous numbers was studied in various contexts.
To avoid confusion, let us make some remarks on the underlying machine model.
The underlying machine model for \ER (over which sentences need to be decided and where reductions are performed) is the \wordRAM (or equivalently, a Turing machine) and not the \realRAM~\citep*{Erickson2022_SmoothingGap} or the Blum-Shub-Smale model~\citep*{Blum1989_ComputationOverTheReals}.

The complexity class \ER gains its importance by numerous important algorithmic problems that have been shown to be complete for this class in recent years.
The name \ER was introduced in~\citet*{Schaefer2010_GeometryTopology} who also pointed out that several \NP-hardness reductions from the literature actually implied \ER-hardness.
For this reason, several important \ER-completeness results had been obtained before the need for a dedicated complexity class became apparent.

Common features of \ER-complete problems are their continuous solution space and the nonlinear relations between their variables.
Important \ER-completeness results include the realizability of abstract order types~\citep*{Mnev1988_UniversalityTheorem,Shor1991_Stretchability} and geometric linkages~\citep*{Schaefer2013_Realizability}, as well as the recognition of geometric segment~\citep*{Kratochvil1994_IntersectionGraphs,Matousek2014_IntersectionGraphsER}, unit-disk~\citep*{Kang2012_Sphere,McDiarmid2013_DiskSegmentGraphs},    and ray intersection graphs~\citep*{Cardinal2018_Intersection}.
More results appeared in the graph drawing community~\citep*{Dobbins2018_AreaUniversality,Erickson2019_CurveStraightening,Lubiw2018_DrawingInPolygonialRegion,Schaefer2021_FixedK}, regarding the Hausdorff distance~\citep*{HausDorff}, regarding polytopes~\citep*{Dobbins2019_NestedPolytopes,Richter1995_Polytopes}, the study of Nash-equilibria~\citep*{Berthelsen2019_MultiPlayerNash,Bilo2016_Nash, Bilo2017_SymmetricNash,Garg2018_MultiPlayer,Schaefer2017_FixedPointsNash}, 
training neural networks~\citep*{Abrahamsen2021_NeuralNetworks, 2022trainFull},
matrix factorization~\citep*{Chistikov2016_Matrix,Schaefer2018_TensorRank,Shitov2016_MatrixFactorizations,Shitov2017_PSMatrixFactorization,tunccel2022computational}, or continuous constraint satisfaction problems~\citep*{Miltzow2022_ContinuousCSP}.
In computational geometry, we would like to mention geometric packing~\citep*{Abrahamsen2020_Framework}, the art gallery problem~\citep*{Abrahamsen2018_ArtGallery}, and covering polygons with convex polygons~\citep*{Abrahamsen2022_Covering}.

Recall that NP is usually described using a witness and a verification algorithm. 
The same characterization exists for \ER.
Instead of the witness consisting of binary words of polynomial length, we allow in addition using real-valued numbers as a witness.
Furthermore, in order to be able to use those real numbers, we are 
allowed to work on the so-called real RAM model of computation.
The real RAM allows arithmetic operations with real numbers in constant time~\citep*{Erickson2022_SmoothingGap}.

\paragraph{Topological Universality.}
Many results and techniques on the existential theory of the reals
actually precede the study of this complexity class.
The underlying idea was to study how complicated solution
spaces can be from a topological perspective.
For example, if we want to study convex polytopes, we are often 
interested in the properties of their face lattice, which is the family of faces of different dimension 
together with their inclusion order
The face lattice is a purely combinatorial object.
Therefore, it is natural to ask which face lattices are
realizable by polytopes.
If there existed an easy combinatorial description
of realizable face lattices, convex polytopes would be much better understood.
Given a specific face lattice $L$, we can study its suitably defined solution space $S(L)$.
As the realizability question can be formulated as an
\ETR-formula, it follows that $S(L)$ is a semi-algebraic set.
Now, let $T$ be a different semi-algebraic set, we wonder
whether there exists a face lattice $L$ such that $S(L)$ is homotopy-equivalent to $T$.
Maybe surprisingly topological universality states that 
there is such an $L$ for any semi-algebraic set $T$.
This type of property feels very strong, as it intuitively states
that we can encode the vast complexity of semi-algebraic sets 
into the problem of realizing convex polytopes.
Indeed, many of the results that establish such topological universality also imply \ER-completeness~\citep*{Matousek2014_IntersectionGraphsER}.

However, topological universality can also be established for \NP-complete problems 
as has been shown by~\citet*{TopologicalArt}.
As they showed, for simplicial complexes, it is sufficient to encode the topology as it is possible to triangulate semi-algebraic sets~\citep*{hironaka1975triangulations}.
In other words, the difference between the \textit{wild} semi-algebraic sets
and the \textit{tame} simplicial complexes is not that the former emit a more complicated
topological structure.
The difference comes from the ability of semi-algebraic sets to encode complicated
topological spaces in a much more \textit{concise} manner.
(To be precise the description complexity of a topological space
might be exponentially smaller using the language of semi-algebraic sets,
compared to simplicial complexes.)
\ER-completeness can be interpreted as giving 
a concise encoding of semi-algebraic sets
into a different domain.
To make our life easier, we do not care about preserving the complete topology, but merely the property of being empty or not.
Still, in order to do so one usually also preserves topological properties.
This is the reason why there is a close connection between
topological universality and \ER-completeness.
But given that \NP-complete problems may also admit
universality theorems \ER-completeness may arguably be considered the 
more interesting finding.

\paragraph{Stronger Universality Results.}
We want to point out that 
previous universality results often showed stronger
results than mere topological universality.
For instance, \citet*{Richter1995_Polytopes} showed such 
a stronger universality theorem for polytopes.
Given a face lattice $F$, we can define the set of polytopes
$V(F)$ having face lattice $F$.
Richter-Gebert showed that for every semi-algebraic set $S$ there exists a 
face-lattice $F$ of a polytope such that $V(F)$ is stably-equivalent to $S$.
To define the notion of stable-equivalence goes above the scope of this paper.
We just note that stable-equivalence encapsulates more than just 
the topology of $S$, but also to a degree the ``geometry'' of $S$.

\subsection*{Discussion on Matroid Input Encodings}

In this paper, we study the $\mathbb R$-realizability, where the input matroid is given with all base, that is, maximal independent sets. This would seem unconventional at first glimpse, especially for those familiar with matroid theory. We would like to address the subtleties around our problem setting.

\paragraph{Types of encodings.} 
For matroids, three possible descriptions are examined in the literature, namely an explicit description of sets, a description via an oracle, and a succinct description with a matrix. 

Recall that an input graph for a graph problem can be given as an adjacency matrix, or equivalently the family of vertex subsets of size two. Similarly, an input hypergraph is typically given as a set family with an explicit description of all hyperedges. An immediate analogue of such an hypergraph description for a matroid is an explicit enumeration of all bases (maximal independent sets) or all circuits (minimally dependent sets). However, explicitly stating all independent sets, bases, or circuits is unconventional;   the most common size measure of a matroid is the number of elements, which is polynomially bounded by the size of a graph or matrix that is generalized by a matroid. 
On the other hand, the number of bases or circuits can be prohibitively large in comparison to the number of elements. Specifically, it is known that the number of distinct matroids on $n$ elements is doubly exponential in $n$~\citep*{KNUTH1974}, hence in space of size polynomial in $n$ one cannot describe an arbitrary input matroid. For further details about explicit matroid encoding, see~\citet*{Mayhew08}.

When the matroids under consideration are representable over a field \F, a matrix over \F\ provides a succinct description of a matroid. 
There are well-studied matroid classes that are representable such as 
uniform matroids, graphic matroids, and transversal matroids. However, not all matroids are representable. Hence, an important question is to decide whether an input matroid is representable over a specific field, or over any field at all, and to find a representation if one exists.

Due to the limitations of the above two explicit descriptions, the most common way to encode an input matroid without any restriction is with an oracle, often an independence oracle. 
One can view an independence oracle as a black box expressing a boolean function on $n$ variables. 
The boolean function is on $n$ input variables and outputs 0 or 1 depending on whether the input corresponds to (a characteristic vector of) an independent set of the said matroid. 
Problems with black box functions, e.g., an implicit input with oracle access are studied in the context of learning a function with a small number of queries, e.g. Polynomial Identity Testing, and also in the context of search problems where a graph is accessed by adjacency query. Such a problem does not fit in the classic computational complexity, where an explicit string of numbers is expected as an input. Moreover, learning a black box (boolean) function cannot be done efficiently. As mentioned previously, even when the boolean input functions are restricted to be matroid oracles, deciding whether a nontrivial matroid property holds or not requires $2^{\Theta(n)}$ queries~\citep*{TRUEMPER1982} even for basic properties such as connectivity and representability over $\mathbb F=GF(2)$, and even with a randomized algorithm. 

Therefore, an oracle encoding of the \F-representability problem does not appear to be a fruitful setting to better understand the algorithmic aspects of the problem. Moreover, we shall argue below that, with an explicit matroid description, there is an intriguing difference in the computational complexity of \F-representability between the cases when \F is finite and when \F=\R.

\paragraph{\F-representability with explicit bases description.} 
Let us consider a matroid description that provides a matroid $M$ 
as a pair $(E,\mathcal B)$, where all bases of $M$ are stated in the collection $\mathcal B$. 
In this setting, we shall argue that the problem of deciding whether $M$ 
is $\mathbb F$-representable is in NP for a finite field $\mathbb F$ and in \ER\ for $\F=\R$. We shall also argue that $\mathbb F$-representability is likely to be in co-NP for a finite field $\mathbb F$. This makes an interesting contrast with the case $\mathbb F = \mathbb R$, for which the corresponding non-representability problem is unlikely to be in \ETR due to our main result.

First, let us see that $\mathbb F$-representability is in NP for finite $\F.$ Indeed, a matrix $A$
(whose columns are labeled by the elements of $E$) 
over \F with $M[A]=M$ can be taken
as a witness for $\F$-representability of $M$. 
Moreover, one can conceive a polynomial-time verification algorithm for the pair $M=(E,\mathcal B)$ and $A$ as follows. 
Let ${\mathcal B}(A)$ be the set of bases of $M[A]$ and recall that $M=M[A]$ if and only if $\mathcal B={\mathcal B}(A)$. 
Whether $\mathcal B\subseteq {\mathcal B}(A)$ can be easily verified in time polynomial in $|\mathcal B|+\rank(M)$. 
For this, we first compute the column rank of $A$. 
If it is different from $\rank(M)$, this trivially implies $M[A]\neq M$. 
Henceforth, let us assume that the column rank of $A$ equals $\rank(M).$ Now, 
for each base $B\in \mathcal B,$ one checks whether the submatrix $A[B]$, the submatrix
of $A$ consisting of all columns labeled by the elements of $B$, is full rank.
If any $B\in \mathcal B$ fails the test, we know that $A$ is not a representation of $M$.

Therefore, we may assume that $\mathcal B\subseteq {\mathcal B}(A)$. To verify whether equality holds, we rely on the following Lemma.
 
\begin{lemma}\label{lem:star}
Let $\mathcal B$ be the set of all bases of a matroid $M$ and let $\mathcal B'\subsetneq \mathcal B.$ Then there exist bases $B'\in \mathcal B'$ and 
$B\in \mathcal B'\setminus \mathcal B$ with $|B'\triangle B|=2.$
\end{lemma}
\begin{proof}
Choose $B'\in \mathcal B'$ and $B\in \mathcal B\setminus \mathcal B'$ so that 
$|B'\cap B|$ is maximized. Let $x\in B'\setminus B.$ Recall the Basis Exchange Property: 
\begin{quote}
    For any distinct bases $W', W$ of a matroid and an element 
    $x\in W'\setminus W$, there exists an element $y\in W\setminus W'$ such that 
    $W'-x+y$ is a basis of the matroid.
\end{quote}
By the Basis Exchange Property, there exists $y\in B\setminus B'$ such that 
$B'':=B'-x+y$ is a basis of $M.$ If $B''$ belongs to $\mathcal B',$ we have  $B''\cap B=(B'\cap B)+y$, which contradicts the choice of $B'$ and $B.$ 
Therefore, $B''$ belongs to $\mathcal B\setminus \mathcal B'.$ Note that $|B'\cap B''|=|B'-x|=rk(M)-1$. Thus $B''=B$ by the choice of $B'$ and $B$ and the claim follows.
\end{proof}

Hence, the last step of the verification algorithm tests if there exist a basis $B\in \mathcal B$ and two elements $x\in B$ and $y\in E-B$ such that $B-x+y$ is not a basis of $M$ but the corresponding set of columns of $A$ is independent, which precisely tests if $B-x+y\in \mathcal B(A)\setminus \mathcal B$. If such a triple $B, x, y$ exists, clearly $\mathcal B\subsetneq {\mathcal B}(A)$ and thus $M[A]\neq M.$ Conversely if $\mathcal B\subsetneq {\mathcal B}(A),$ there exists such a triple $B, x, y$ by \Cref{lem:star}. Therefore, we examine all triples $(B,x,y)$ with $B\in \mathcal B$, $x\in B$ and $y\in E-B$ and certify that $B-x+y$ is either in $\mathcal B$ or dependent, in which case the verification algorithm can correctly conclude that $\mathcal B(A)=\mathcal B$, and thus $M[A]=M.$ Otherwise, the verification algorithm concludes $M[A]\neq M$ and rejects the witness $A$.

The presented verification algorithm shows that $\F$-representability is in NP for each finite field $\F.$ A matrix over $\mathbb F$ as witness and the polynomial-time verification algorithm naturally extend to the case when $\mathbb F=\mathbb R$, 
where the verification algorithm works on a real RAM. Therefore, $\mathbb R$-representability is in $\ER$. For details on the 
characterization of $\ER$ via a witness and a polynomial-time verification algorithm on a real RAM, see~\citet*{Erickson2022_SmoothingGap}, and also the discussion in the paragraph above about the existential theory of the reals.

Furthermore, it is likely that $\F$-representability is in co-NP for each finite field $\F.$ 
It is known~\citep*{GeelenW16} that for any prime field $\F$, non-$\F$-representability can be certified by evaluating the ranks of $O(n^2)$ subsets of an $n$-element matroid. 
Notice that when the matroid $M=(E,\mathcal B)$ is given with an explicit description of all the bases $\mathcal B$ of $M$, evaluating $\rank(X)$ for $X\subseteq E$ can be done in time polynomial in the input size because $\rank(X)$ equals the maximum of $|X\cap B|$ over all $B\in \mathcal B.$ 
Therefore, $\F$-representability is in co-NP under the explicit bases description for each prime field $\F.$ For an arbitrary finite field $\F$, not necessarily prime, up to our best knowledge there is no published result which establishes that a polynomial number of rank evaluations suffices for non-$\F$-representability. 
However, it is known that a positive resolution of Rota's conjecture implies that only a constant, depending on $|\F|$ only, number of rank evaluations would suffice~\citep*{oxley2006matroid} to certify that a given matroid is not $\F$-representable. 
The proof of Rota's conjecture was announced in 2014 by~\citet*{GGW14} although it is expected to take a few more years for the full proof to be written for publication.

Therefore, $\F$-representability appears to be in NP $\cap$ co-NP when the input is given as the exhaustive list of bases for each finite $\F$ given the claimed proof of Rota's conjecture. 
Given that, deciding the computational complexity of $\F$-representability for each finite $\F$ with explicit bases description is an intriguing question. For $\F=GF(2),$ a polynomial-time algorithm is straightforward from the uniqueness of a binary representation (up to linear transformation) and the fact that such a representation can be efficiently obtained~\citep*{oxley2006matroid}; after constructing a matrix over $GF(2)$, we apply the above verification algorithm for NP membership. However, even for $\F=GF(3)$ it is not clear whether a matrix over GF(3) can be efficiently constructed although it is known that there is a unique representation over $GF(3)$ for a matroid representable over $GF(3).$ As far as we are aware, there is no efficient procedure known for constructing the representation of a matroid $M$ with a promise that $M$ is representable over $GF(3),$ when $M$ is given with an independence oracle or even given as a matrix over the rationals $\mathbb Q$~\citep*{Hlineny06}. Getting an input matroid as explicit bases description might help to circumvent this obstacle.

In contrast to the case of finite fields, it is impossible to certify non-$\R$-representability with a polynomial number of rank evaluations~\citep*{oxley2006matroid}. Finally, for $\R$-representability we showed \ER-completeness, exhibiting a noticeable diversion from $\F$-representability for finite $\F$ which appears neither NP-complete nor co-NP-complete with explicit bases description under the assumption NP $\neq$ co-NP. Therefore, our result highlights the recurring contrast between representability over a finite field and over the reals.

\section{Distinct-ETR}
\label{sec:Distinct}
This section serves as a preparation for the later reduction. Specifically, we show in Lemma~\ref{lem:distinctER} the \ER-completeness of the variant of \ETR called \Distinct that we defined in the introduction: in this variant, we are guaranteed that either there is no solution or there is a solution with all variables holding \textit{distinct} values. 
This property will be key for encoding into matroids. Most of this section follows standard techniques.

\paragraph{Overview.}
This section is dedicated to the proof of the lemma.
The idea is that we first establish the hardness of an \ETR variant (\Strict ) with an open solution space.
It is clear that all variables can be assumed there to have distinct values.
Then, we reduce again to a variant where we use only the basic constraints.  

The reduction goes in four steps.
\begin{enumerate}
    \item From \ETR to \ETRAMI.
    \item Then from \ETRAMI to \Feasibility.
    \item Then from  \Feasibility to \Strict.
    \item And at last from \Strict to \Distinct.
\end{enumerate}

Note that steps 1, 2, and 3 have already been done (among others) by~\citet*{Schaefer2017_FixedPointsNash}.
We sketch the main steps of their reduction.
Specifically, we point out some properties that were not explicitly emphasized.
We start to explain a simple trick that is excessively used in those types of constructions in order to build small, very small, and  very large numbers.

\paragraph{Number Constructions.}
Before we describe the reduction, we show how to construct variables that must have specific rational values.

If we want to build integers of polynomial size, we can do this by simply repeatedly adding a one.
\[a_1 = 1,\quad  a_{i+1} = a_i + a_1.\]
It is easy to see that $a_i = i$, for all $a_i$ that are defined in this way.

If we want to build a very large number, say $2^{2^{k}}$, the previous approach cannot be done in a polynomial number of steps.
Instead, we can use repeated squaring as follows.
\[x_0 = a_2 + a_0 (= 2), \quad x_{i+1} = x_i^2,\]
for $i = 1,\ldots,k$.
It holds inductively that $x_i  = 2^{2^{i}}$.
Similarly, we can construct very small numbers  say $2^{2^{-k}}$, as follows:
\[y_0 + y_0 = 1, \quad y_{i+1} = y_i^2,\]
for $i = 1,\ldots,k$.
It holds inductively that $y_i  = 2^{-2^{i}}$.
Note that we can also use strict inequalities to build large and small numbers.
For example,
\[x_0 > 2, \quad x_{i+1} > x_i^2,\]
for $i = 1,\ldots,k$ 
implies that $x_i  > 2^{2^{i}}$.

We will use these standard tricks repeatedly later in the reduction.

\paragraph{Reduction from \ETR to \ETRAMI.}
We define the problem \ETRAMI (see~\citet*{dynamic} for more background on this problem) as a variant of \ETR as follows.
We are given a set of variables $X = \{x_1,\ldots,x_n\}$ and constraints of the form 
\[x+y = z, \quad  x\cdot y = z, \quad x = 1, \]
for $x,y,z\in X.$ Therefore, \ETRAMI is a variant of \ETR without negations, inequalities, disjunctions, conjunctions and the only constant is $1$.
It is folklore that \ETRAMI is \ER-complete and follows implicitly from various papers, for example~\citet*{Matousek2014_IntersectionGraphsER, Schaefer2017_FixedPointsNash, Shor1991_Stretchability}.
\begin{lemma}[folklore]
    \ETRAMI is \ER-complete.
\end{lemma}
\ER-membership follows from the definition.
The idea of the reduction is to simplify an \ETR-formula in each step. 
For instance, we can remove negations by replacing $\lnot p>0$ by $p\leq 0$, where $p$ is a polynomial. In this way, we can remove all negations.
We can replace inequalities by observing that $p\geq0$ is equivalent to $\exists x\colon p = x^2$  and 
$p>0$ is equivalent to $\exists x: px^2 = 1$.
We can replace $p=0 \land q=0$ by $p^2 + q^2 = 0$.
Similarly, we can replace $p=0 \lor q=0$ by $pq = 0$.
Thereafter, we end with a single polynomial equation $p=0$.
We construct variables for each coefficient value.
Then we replace each coefficient with an appropriate variable.
At last, we replace each occurrence of multiplication and addition inductively by introducing one more variable and one more constraint. 
We end with a single equation of the form $x=0$, which can be replaced by $z= 1$ and $z + x = z$.

\paragraph{Reduction from \ETRAMI to \Feasibility.}

In \Feasibility, we are given a single polynomial $p \in \Z[x_1,\ldots,x_n]$ of degree at most four.
We are asked if there exists some $x\in \R^n$ such that
$p(x) = 0$.
Furthermore, we require each coefficient to be of absolute value at most $36n^3$. Below, we will show how to achieve this upper bound.

\begin{lemma}
    \Feasibility is \ER-complete.
\end{lemma}
Again, \ER-membership follows from the fact that \Feasibility is a special case of \ETR.
To show hardness we sketch a reduction from \ETRAMI that is already known~\citep*{Matousek2014_IntersectionGraphsER, Schaefer2017_FixedPointsNash}.
Let $f_1=0,\ldots,f_m=0$ be the equations of some \ETRAMI instance $\varphi$.
(For example $x+y=z$ becomes $x+y-z = 0$.)
Let 
\[ p = f_1^2 + \ldots + f_m^2.\]
Clearly, $\varphi$ is satisfiable if and only if $p$ has a zero. 
As each $f_i$ has a degree at most two, it holds that $p$  has degree at most four.
Since there are three types of equations, each involving at most three variables, there are only $3n^3$ possible distinct constraints in $\varphi$. Thus we have $m\leq 3n^3$. 
Note that each term $f_i^2$ gives rise to at most six monomials and each coefficient is at most two. (For example, $(x+y-z)^2 = x^2 + 2xy - 2xz + y^2 -2yz + z^2$.)
Therefore each coefficient has absolute value at most $12m = 36n^3$.
Thus, we can rewrite $p$ as a sum of monomials with bounded-sized coefficients, as claimed.

\paragraph{Reduction from \Feasibility to \Strict.}
In a \Strict instance, we are given a sentence of the form
\[
    \exists x_1, \ldots, x_n \in \R :
    \varphi(x_1, \ldots, x_n),
\]
where~$\varphi$ is a well-formed and quantifier-free formula consisting of polynomial strict-inequalities in the variables and the logical connectives $\{\land, \lor\}$.

Note that the solution space $\{x \in \R^n: \varphi(x)\}$ is always open.

\begin{lemma}
    \Strict is \ER-complete.
\end{lemma}
Again membership follows from the fact that \ETR is more general then \Strict.
To show \ER-hardness we reduce from \Feasibility.
The idea of the reduction is to replace $p(x) = 0$ by $-\delta <p(x)< \delta$, for some sufficiently small $\delta$.
To this end, we employ two lemmas as formulated in~\citet*{Schaefer2017_FixedPointsNash}.
Note that the actual proof comes from real algebraic geometry and can be found for instance in~\citet*{basu2010bounding} and~\citet*{jeronimo2013minimum}.

\begin{lemma}
    \label{lem:BoundingRadii}
    Every non-empty semi-algebraic set in $\R^n$ of complexity at most $L \geq 4$ contains a point of distance at most $2^{L^{8n}}$ from the origin.
\end{lemma}

\begin{lemma}
\label{lem:MinDistance}
    If two semi-algebraic sets in $\R^n$ each of complexity at most $L \geq 5n$ have positive distance (for example, if they are disjoint and compact), then that
distance is at least $2^{-2^{L+5}}$ .
\end{lemma}
In this context, the distance between two sets $A$ and $B$ in $\mathbb{R}^n$ is defined as 
\[d(A,B) = \inf_{a\in A, b\in B} \|a-b\|,\]
where that $\|\cdot\|$ denotes the Euclidean norm.

In order to apply \Cref{lem:BoundingRadii,lem:MinDistance}, we define the following three semi-algebraic sets. 
First, we define the solution set for $p$.
\[S = \{ x\in \R^n : p(x) = 0 \}.\]
Let $R = 2^{L^{8n}}$, where $L$ is the bit-complexity of $S$.
Note that $L$ upper bounded by the length of $p$.
Using \Cref{lem:BoundingRadii}, we know that $S$ is empty if and only if $S\cap B(R)$ is empty.
(We denote by $B(R)$ the ball of radius $R$ around the origin.)
This motivates us to define the sets
\[S_1' = \{ (x,z)\in \R^{n+1} : p(x) = z \land  \|x\|^2 \leq R^2\},\]
and 
\[S_2' = \{ (x,z)\in \R^{n+1} :  z = 0 \land  \|x\|^2 \leq R^2\}.\]
Note that $S_1'\cap S_2' = (S\cap B(R))\times \{0\}$.
Furthermore, $S_1'$ and $S_2'$ are compact and thus we can apply \Cref{lem:MinDistance}.
Unfortunately, the description complexities of $S_1'$ and $S_2'$ are exponential if we write $R$ out in binary.
Therefore, we define $S_1$ and $S_2$ slightly differently. 
Namely, we add some extra variables, whose sole purpose is to encode $R$ using repeated squaring. 
Let $\overline{L}$ be the maximum of the bit complexities of $S_1$ and $S_2$. Note that $\overline{L} = O(L + n \log L )$.
Let $\delta$ be a lower bound on the distance between $S_1$ and $S_2$, which is provided by \Cref{lem:MinDistance} applied to $S_1$ and $S_2$.
It holds that $S = \emptyset$ is equivalent to $S\cap B(R) = \emptyset$.
This in turn is equivalent to $S_1\cap S_2 = \emptyset$.
And this is equivalent to
\[p(x) \leq -\delta   \text{ or } \delta\leq p(x),\]
for all $x\in B(R)$.
In other words, we have
\begin{equation} \label{eq:1}
\exists x: \ -\delta <  p(x)<\delta \text{ and }  \|x\|^2 < R^2
\end{equation}
if and only if
\begin{equation*}
    \exists x: \ p(x) = 0
\end{equation*} 
This obviously also works if we would use \textit{any} smaller $\delta$.
Maybe not so obviously, this also works for any $R$ between $2^{L^{8n}}$ and $2^{L^{8n+1}}$.
The lower bound allows us to apply \Cref{lem:BoundingRadii}.
The upper bound allows us to apply \Cref{lem:MinDistance}.
Using repeated squaring, we create numbers $a > 2^{L^{8n}} $ and $b < 2^{L^{8n+1}}$. Furthermore, we add the inequalities $a<R<b$.
Note that we can construct a number $\delta$ that is at most $2^{-2^{\overline{L}+5}}$.
Our \Strict instance $\varphi$ consists of the three inequalities from \Cref{eq:1} and some extra variables and constraints to bound $R$  and $\delta$ as described above.

\paragraph{Reduction from \Strict to \Distinct.}
We can now finally show that \Distinct is \ER-complete.
\begin{lemma}\label{lem:distinctER}
    \Distinct is \ER-complete
\end{lemma}

\begin{proof}
We are not actually reducing from \Strict but from the instance $\varphi$ described in \Cref{eq:1}. 
Let $\delta, R$ be the given numbers, $x_1,\ldots,x_n$ the variables and $p$ the polynomial as in the previous paragraph.
Recall that $p$ has a degree at most four and the coefficients are bounded integers.
We will introduce new variables in order to construct $\delta$, $R$, $\|x\|^2$, and $p(x)$. 
Then, we will argue about distinctness.

First, we construct variables holding the values of the integers $-36n^3,\ldots, 36n^3$.
Those variables are meant to represent the coefficients of $p$.
Recall that the values of those coefficients were in this range.
It has been described above how to construct those integers.
Furthermore, we add variables holding the values $R = 2^{L^{8n}}$ and $\delta = 2^{-2^{\overline{L}+5}}$. 
(As we reduce from \Cref{eq:1}, we do not need to approximate the values $\delta $ and $R$, but can construct them directly, as \Distinct allows us to use equations.)
 If in this process two variables hold the same value, we can detect this and remove one of the variables.

Now, we construct $p(x)$ and $\|x\|^2$.
First, we construct all possible $\binom{n}{4}$ monomials of degree at most four.
For example, $N = xyzw$ is constructed in three steps.
$N_1 = xy$, $N_2 = N_1z$, and $N = N_2 w$.
Again, whenever two identical monomials appear,
we are able to notice this and make an appropriate replacement.
Let us denote 
\[p(x) = \sum_{i=1}^m a_i M_i.\]
Here, $a_i$ is the coefficient of the monomial $M_i$.
We construct $P_k = \sum_{i=1}^k a_i M_i$ inductively as follows:
\[P_1 = a_1 M_1 \text{ and } P_{k} = P_{k-1} + T_k,\]
with $T_k = a_kM_k$.

We denote by $P= P_m$ the variable holding the value of $p(x)$.
We construct $X = \|x\|^2 = x_1^2 + \ldots + x_n^2$ in the same way. 

At last, we add the variables $a,b,c$, and the constraints
\[P + \delta = a, \quad b + P = \delta, \quad R_2=R \cdot R, \quad c + X = R_2.\]
Now we can enforce the inequalities 
\[
 \ -\delta <  p(x)<\delta \text{ and }  \|x\|^2 < R^2
\]
by \[a>0,\quad b>0,\quad c > 0.\]

To summarize, 
we started with the variables $x_1,\ldots,x_n$.
We have created variables $C_1,\ldots, C_s$ that each holds a different integer/rational number.
Furthermore, we constructed some variables $V_1,\ldots, V_t$ such that each $V_i$ is a polynomial function $g_i(x)$.
Note that all the $g_i$ have a degree of at most four and small integer coefficients.
If for two variables $V_i$ and $V_j$, we have that $g_i = g_j$, we can detect this and remove one of them and replace each occurrence with the other one.
This finishes the description of the reduction.
We denote this instance $\psi$.

\bigskip

To show correctness, we observe that $\varphi$ has a solution if and only if
 $\psi$ has a solution as well.
Indeed, all new variables and constraints only ``build''  the correct polynomials and the numbers $\delta$ and $R$.
Thus it remains to show that $\varphi$ has a solution if and only if $\psi$ has a solution with all variables taking distinct values.
The backward direction is trivial.
Therefore, we assume that $\varphi$ has at least one solution $x\in \R^n$. As the solution  space of $\varphi$ is open, there is an open ball $B$ fully contained in the solution space.
Clearly, the variables $C_1,\ldots,C_s$ have all fixed distinct values by construction.
Every other variable $V_i$ can be expressed as a polynomial function $g_i(x)$. As all the $g_i$ are distinct, there must be some $x\in B$ such that $g_i(x)\neq g_j(x)$, for all $i\neq j$ and $g_i(x)\neq C_j$, for all $i,j$.
Otherwise, two of the polynomials would be identical.
This finishes the proof.
\end{proof}

\section{Arithmetic Using Matroids}
\label{sec:Arithmetic}

In this section, we describe how to encode addition and multiplication of real numbers using rank-$3$ matroids. We rely on the \emph{von Staudt constructions}, which are very well-known, perhaps with the caveat that they are usually stated for oriented matroids. But as we shall see, no orientedness is actually required to make them work. We follow the presentation of~\citet*{Matousek2014_IntersectionGraphsER}.

The setup for both operations is as follows. 
We have a line $\ell$ containing three distinct distinguished points called $\0$, $\1$, and $\II$, and a fixed second line $\ell_\infty$ crossing $\ell$ at $\II$.
Given the variable $x$, we denote the corresponding point representing it by $\x$ in boldface. In this way, we easily distinguish between a point and the corresponding variable.
A point $\x$ on the line $\ell$ can be interpreted as a real number using cross-ratios: if we denote by $d(\a,\b)$ the oriented distance between two points $\a$ and $\b$ on the line $\ell$, then the quantity \[(\x,\1;\0,\II):=\frac{d(\x,\0)\cdot d(\1,\II)}{d(\x,\II)\cdot d(\1,\0)}\] is a real number invariant under projective transformations, which, by a slight abuse of notation, we simply denote by $x$. 
Note that if $\II$ is progressively sent to infinity using a projective transformation and $d(\0,\1)$ is scaled to one in this formula, $x$ converges to $d(0,x)$.
Thus this cross-ratio matches the geometric location of $\x$ on $\ell$ under some projective transformation.

Now, given two points $\x$ and $\y$ on $\ell$, we describe geometric operations to compute points $\boldsymbol{x+y}$ and $\boldsymbol{x \cdot y}$ on $\ell$ representing their addition and their multiplication.

\begin{figure}[ht]
    \centering
    \def\svgwidth{\textwidth}
    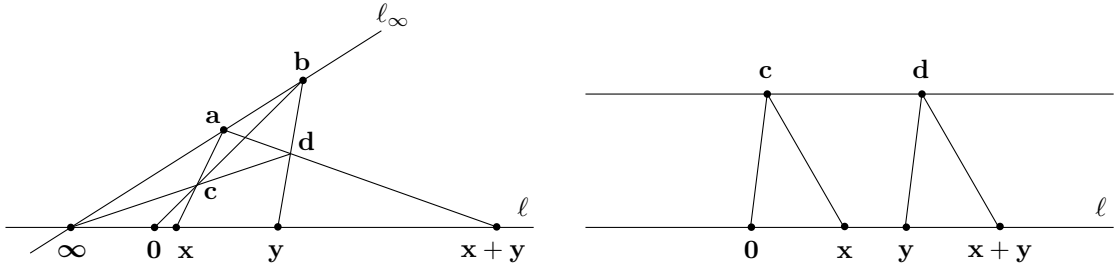
    \caption{Encoding addition geometrically.}
    \label{F:addition}
\end{figure}

\paragraph{Addition.} The construction is shown in Figure~\ref{F:addition}, left. We first introduce two distinct auxiliary helper points $\a$ and $\b$ situated anywhere on $\ell_\infty$. 
The line connecting $\0$ to $\b$ crosses the line connecting $\x$ to $\a$ in a point $\c$, then the line connecting $\II$ to $\c$ crosses the line connecting $\y$ to $\b$ in a point $\d$.
Similarly, the line connecting $\a$ to $\d$ crosses $\ell$ in a point that we define to be $\boldsymbol{x+y}$. 
The rationale behind this construction is that by a projective transformation we can consider $\ell_{\infty}$ to be a line at infinity, leading us to Figure~\ref{F:addition}, right. 
Now the line containing $\c$ and $\d$ crosses $\ell$ at a point at infinity, i.e., these two lines are parallel. 
Likewise, the line containing $\c$ and $\d$ is parallel to $\ell$. Then the parallelity of the lines $\0\c\b$ and $\y\d\b$ as well as $\x\c\a$ and $(\boldsymbol{x+y})\d\a$ immediately shows that $d(\0,\boldsymbol{x+y})=d(\0,\x)+d(\0,\y)$. 
In other words, the point $\boldsymbol{x+y}$ has value $x+y$, justifying the notation.

\paragraph{Multiplication.} The construction is shown in Figure~\ref{F:multiplication}, left. 
As before, we first introduce two distinct auxiliary helper points $\a$ and $\b$ situated anywhere on $\ell_\infty$. 
The line connecting $\1$ and $\b$ crosses the line connecting $\x$ and $\a$ at a point $\c$, and the line connecting $\0$ and $\c$ crosses the line connecting $\b$ and $\y$ at a point $\d$. 
Finally, the line connecting $\a$ and $\d$ crosses the line $\ell$ at a point that we define to be $\boldsymbol{xy}$. 
By sending the line $\ell_{\infty}$ at infinity using a projective transformation, we obtain Figure~\ref{F:multiplication}, right, where one can readily show using the parallel lines that $d(\0,\boldsymbol{xy})=d(\0,\x)d(\0,\y)$.
In other words, the point $\boldsymbol{xy}$ has value $xy$, justifying the notation.

\begin{figure}[ht]
    \centering
    \def\svgwidth{\textwidth}
    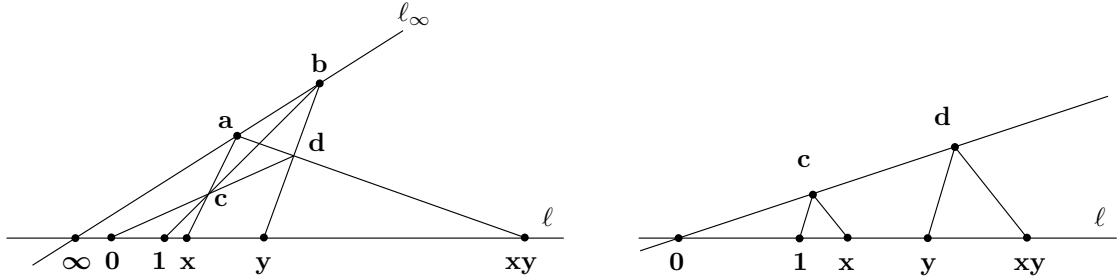
    \caption{Encoding multiplication geometrically.}
    \label{F:multiplication}
\end{figure}

\paragraph{Encoding these geometric constructions using matroids.} The addition and multiplication constructions defined above can be entirely encoded using matroids: the independent sets are exactly the empty set, all the singletons, all the pairs of points, and all the triples of non-aligned points. We point out that no orientation was ever enforced during the constructions, and thus we do not need oriented matroids to describe them. Furthermore, these operations can be chained arbitrarily, allowing to encode polynomials using matroids. However, some very important care needs to be taken here: while it is clear from the construction which points should be aligned, we also need to make sure that points that should not be aligned can be assumed to not be aligned. For example, it could a priori happen that a line going through two helper points $c_1$ and $c_2$ somehow accidentally happens to pass through a variable $x$ of the polynomial we are encoding. In that case, the matroid would not properly encode the geometric construction.

This motivates the following definition: we say that the set $S$ of helper points used during an addition or multiplication construction is \emph{free} if for any finite set of lines $L$ and points $P$ not involved in the construction, the points in $S$ can be perturbed so that:
\begin{itemize}
\item the incidences required by the addition or multiplication construction still hold,
\item no point of $S$ lies on a line $L$, and
\item no pair of points of $S$ is aligned with a point of $P$. \\
(In particular, no point of $S$ coincides with a point of $P$.)
\end{itemize}

A key property of the addition and multiplication construction is that the four helper points that they rely on form a free set. Indeed, the points $a$ and $b$ can be placed freely on the line $\ell_{\infty}$, and thus can be perturbed so as to avoid the lines and points of $L$ and $P$. Such a perturbation induces a perturbation of $c$ and $d$ in a two-dimensional open set, therefore allowing them to avoid lines of $P$ and $L$, but also ensuring that no line going through a pair of points in $\{a,b,c,d\}$ also goes through a point in $P$.

This freedom will be leveraged in the proofs of \Cref{thm:ordertypes} and \Cref{thm:Self-Contained}, to ensure that the matroid correctly encodes orientation predicates and systems of polynomial equations, respectively.

\paragraph{Strict inequalities.} The multiplication construction can be leveraged to simulate a strict inequality constraint $x>0$. Indeed, $x>0$ if and only if there exists $z\in \mathbb{R}$ such that $z\neq 0$ and $x=z^2=zz$, which can thus be simulated using a helper point $z$ distinct from $0$ and the above multiplication construction. However, this construction would require us to use as a helper point a fixed point $z$ on the line $\ell$, which could not be perturbed and thus would not be free. This can be resolved by using an additional variable: we first introduce a helper point $y$ for which we ensure that $y>0$ using the multiplication gadget. Then we use another multiplication gadget to ensure that $z>y$: note that this amounts to enforcing that $z$ lies on the same side of $y$ as $1$ does, i.e., this can be tested by using another multiplication gadget where $0$ is replaced by $y$. Now, in these two multiplication gadgets, neither $y$ nor any of the other helper points is fixed, and therefore we can perturb them to avoid any fixed set of lines and points, showing that they form a free set of points.

\section{Proof of \Cref{thm:ordertypes}}
\label{sec:Proof2}

For convenience, we first restate \Cref{thm:ordertypes}

\ordertypes*

\Cref{thm:ordertypes} is proved by using the arithmetic constructions described in the previous section, in particular the one for strict inequalities, to simulate orientation predicates. 

\paragraph{Simulating rank-3 order types.}

We first consider the case of rank equal to $3$ and treat the general case later.
Let $O = (E,\chi)$ be an order type on $n$ elements of rank $3$.
(We assume that $E = \{e_1,\ldots,e_n\}$.)
We construct a matroid $M$ from $O$ inductively.
To be more precise, we construct matroids, $M_3,\ldots,M_n = M$
such that $M_i$ simulates $O_i$.
Here, $O_i = (E_i, \chi_i)$ is the  order type formed by the first $i$ elements of $O$. All of our matroids $M_i$ will feature a distinguished line $\ell_{\infty}$.

The matroid $M_3$ merely contains 
$e_1,e_2,e_3$.
Without loss of generality we assume that the triple $t=\{e_1,e_2,e_3\}$  is 
independent in $O$.
(Otherwise, all triples in $O$ are dependent, which can trivially be simulated.) We first add a line at infinity, on which none of $e_1$, $e_2$ or $e_3$ lies.
We can assume that $t$ is oriented correctly in any representation of $M_3$,
as otherwise, we can just reflect the representation and
get a correct representation.
Therefore, $M_3$ satisfies the induction hypothesis.

Now, let us assume that we have already constructed $M_{i-1}$. 
We construct $M_i$ from $M_{i-1}$, by adding the element $e = e_i$ from $O$.
Furthermore, for each triple $t = \{a,b,e\} \subset E_i$, we 
need to ensure that $t$ is oriented correctly.
In case that $\chi(t) = 0$, we can encode this directly by specifying that $\{a,b,e\}$ forms a rank-$2$ dependent set in the matroid $M_i$.
Thus it remains to consider the case $\chi(t)\neq 0$.
Let $t' = \{a,b,c\} \subseteq E_{i-1}$ be such that
$\chi(t' ) \neq 0$.
Note that such a triple $t'$ must exists, as otherwise all points of $E_i$
lie on a common line. 
This would be a contradiction to the fact that $M_3$ is formed by an independent triple.
Using the orientation of $t'$, we add a small constant number of helper points
and we will enforce the correct orientation of $t$ in $M_i$ as well.
Thus it remains to show the following lemma, where we use the notion of a free set of helper points defined in \Cref{sec:Arithmetic}.
\begin{lemma}\label{lem:orientationgadget}
    Given the independent triple $t' = \{a,b,c\}$ in a matroid $M$
    and another point $e$, we can enforce that in any valid representation of $M$, the triple $t = \{a,b,e\}$
    is oriented identically to $t'$, or that it is oriented opposite to $t'$.
    We do this by adding a constant number of helper points to $M$. Furthermore, this set of helper points is free.
\end{lemma}

\begin{proof}
This construction goes in essentially two steps.
First, we show how we can enforce an element $x$
on the line $\ell = \ell(a,b)$
to be on the same side of $a$ as $b$.

\begin{center}
    \includegraphics[page = 2]{Force-Orienation.pdf}
\end{center}

This first step relies on standard constructions to encode arithmetic operations as explained in
\Cref{sec:Arithmetic}.
Although those constructions can be well described
and understood, without any reference to arithmetic operations,
the language of arithmetic operations gives
a better intuition.
The underlying idea is that we interpret $a$ as zero, and $b$ as one. Indeed, the constraint that $x$ lies on the same side of $a$ as $b$  in any representation where the line $\ell_{\infty}$ is sent to infinity amounts to enforcing that $x>0$ when interpreting $a$ as zero and $b$ as one. As explained in \Cref{sec:Arithmetic}, this can be encoded using multiplication gadgets, in such a way that none of the helper points accidentally lie on a previously used line or point.

In the second step, we use the previous tool to enforce that
$e$ lies on the same (or the opposite) side  of the line $\ell(a,b)$ as $x$.

\begin{figure}[ht]
    \centering
    \def\svgwidth{6cm}
    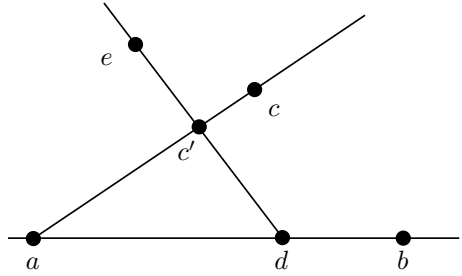
    \caption{Forcing $c$ to be on the same side of $\ell(a,b)$ as $e$. Rays ending at a point denote the strict inequality constraints.}
    \label{F:triples}
\end{figure}

To this end, we first define a point $c'$ for which we enforce that on the line $\ell(a,c)$, it lies on the same side of $a$ as $c$. Then we define a point $d$ situated on the line $\ell(a,b)$ and
the line $\ell'$ through $e$ and $c'$, see Figure~\ref{F:triples}
The condition that $e,c'$ are on the same side of $d$ on the line $\ell'$ can be enforced using the previous gadgets. This condition is equivalent to $c'$ and $e$ being on the same side of $\ell(a,b)$. Lastly, by construction the triple $\{a,b,c'\}$ is oriented identically to $\{a,b,c\}$. 

We can use the same tool to enforce that $e$ and $c$ lie on the opposite side of the line $\ell(a,b)$. We define $c'$ as before, so that it lies on the same side of $\ell(a,b)$ as $c$, and then we simply need to enforce instead that $c'$ lies on the opposite side of $\ell'$ with respect to $d$ (where $\ell'$ and $d$ are defined as above).

As explained in \Cref{sec:Arithmetic}, the strict inequalities gadgets can be directly encoded into the independent sets of the matroid, and by construction, the helper points always have at least one degree of freedom, and thus can form a free set of points. Therefore, the matroid $M_i$ is entirely defined by the dependency constraints indicated by the lines in the geometric constructions. This concludes the proof.
\end{proof}

We can now conclude the proof of \Cref{thm:ordertypes} for rank-$3$ matroids. Given an  order type $O$ of rank $3$, for which we can assume that there is at least one independent set of size $3$ (otherwise the orientation predicates are trivial), we inductively encode the orientation predicates into a matroid using the helper points provided by Lemma~\ref{lem:orientationgadget}. At each stage of the induction, the freedom of the set of helper points can be used to ensure that the helper points do not yield any accidental dependencies with all the points and lines previously placed. At the end of the induction, by Lemma~\ref{lem:orientationgadget}, any valid representation of the resulting matroid $M$ induces a representation of $O$. Conversely, any representation of $O$ can be extended to a representation of $M$. Therefore $M$ simulates $O$. All the constructions can clearly be done in linear time, which concludes the proof.

\paragraph{Simulating Rank-$k$ Matroids.}
The construction for rank $k$ is identical to that of rank~$3$ as explained above,
with two exceptions.
First, the induction basis starts with a matroid on $k$ elements instead of $3$.
Second, we have to describe the simulation of an oriented $k$-tuple.
For this, we use the same trick that forces a point to lie on a 
specific side of a line, where we just replace a line with a hyperplane~$P$, as shown in Figure~\ref{F:hyperplane}.

\begin{figure}[ht]
    \centering
    \def\svgwidth{7cm}
    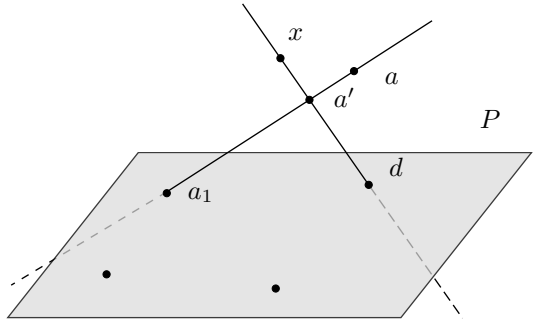
    \caption{Forcing $x$ to be on the same side of $P$ as $a$.}
    \label{F:hyperplane}
\end{figure}

More precisely, in an inductive step where we add a point $x$, we need to enforce the orientation of all the independent $k$-tuples $t=\{a_1, \ldots a_{k-1}, x\}$. For this, we take another $k$-tuple $t'=\{a_1, \ldots a_{k-1}, a\}$ for which $\chi(t')\neq 0$, and we devise a gadget to encode the constraint that $t$ is oriented identically, or opposite to $t'$. This will be done using the gadget described above that enforces that in any valid representation, for three aligned points $a, b$ and $c$, $c$ lies on the same side of $a$ as $b$. That gadget was defined for rank $3$ matroids and thus representations into $\mathbb{R}^2$, but readily works in higher dimensions: one must simply ensure using dependency constraints that all the points involved in the gadget lie on a common plane. The hyperplane at infinity will intersect this common plane in a line, which takes the role of the line at infinity in the gadgets.

Now, we proceed as in the rank $3$ case: considering the hyperplane $P$ generated by $\{a_1, \ldots a_{k-1}\}$, 
we first define a point $a'$ that lies on the same side as $a$ on the line $\ell(a_1,a)$. We introduce the point $d$ at the intersection of $P$ and the line going through $a'$ and $x$. Then we enforce that $x$ is on the same side of $d$ as $a'$. In order to enforce that $a'$ and $x$ are on opposite sides of $P$, we instead enforce that $x$ and $a'$ are on opposite sides of $d$.
Finally, we observe that all the helper points that we have introduced are free, where the notion of freedom is generalized to also disallow $k$-dimensional dependencies: once again this follows from the fact that all the helper points that we introduce have some wiggle room to be perturbed. The rest of the proof proceeds identically to the rank-$3$ case.

\section{Proof of \Cref{thm:Self-Contained}}
\label{sec:Proof1}

For convenience, let us first restate \Cref{thm:Self-Contained}

\selfcontained*

The proof of \Cref{thm:Self-Contained} is by a direct reduction from the problem \Distinct, using the arithmetic constructions described in \Cref{sec:Arithmetic}. Let us start with an instance of \Distinct given by variables $X= \{x_1, \ldots, x_n\}$ and constraints of the form

\[x+y = z, \quad  x\cdot y = z, \quad x = 1, \quad x>0,\]

for $x,y,z \in X$, with the promise that if there is a solution, then there is one where all the variables are pairwise different.

As in the proof of \Cref{thm:ordertypes}, we can construct the constant $0$ using a constraint $x+y=x$, and thus, without loss of generality, by the distinctness assumption, we can assume that there are exactly two variables in $X$ equal to respectively $0$ or $1$, and that the others are different from $0$ and $1$. By a slight abuse of language, we remove those from $X$ and denote them directly by $0$ and $1$ in the rest of the description. We define a rank-$3$ matroid $M$ as follows. First, we have a line $\ell$ consisting of three distinguished distinct points $0$, $1$ and $\infty$, as well as $n$ distinct points (and distinct from $\{0,1,\infty\}$) corresponding to the variables $\{x_1, \ldots x_n\}$. We also add a line at infinity $\ell_\infty$ going through $\infty$ and no other point.

We then use the gadgets from \Cref{sec:Arithmetic} to inductively encode all the constraints. Note that there is at most one constraint $x=1$ which can be hardcoded from the start. Then let us assume inductively that we have defined a matroid $M_i$ encoding the first $i$ constraints. The $(i+1)$-th constraint is an addition, a multiplication, or a strict inequality, which can be encoded using a geometric construction as described in \Cref{sec:Arithmetic}. Now, the key property of these constructions is that the set of helper points is free. Therefore, the only linear dependencies involved in the construction are those of that construction, which can be readily encoded into a matroid $M_{i+1}$. 

We now prove that the \Distinct instance has a solution with distinct variables if and only if the matroid $M$ is representable over the reals. First, if \Distinct has a solution, we obtain a representation of $M$ over the reals by placing all the variables $x_1, \ldots ,x_n$ on the line $\ell$ at the values indicated by the solution, and by sending the line at infinity to infinity. The geometric constructions are then represented one by one, and since the helper points are free, by perturbing them if needed we can ensure that they are all distinct, that no three of them are colinear, and that they do not form colinearities with previously placed points. Therefore, this constitutes a correct representation of the matroid. Conversely, given a representation of the matroid $M$ over the reals, we read the values of the variables on the line $\ell$ using cross-ratios as explained in \Cref{sec:Arithmetic} (or equivalently we send $\ell_\infty$ to $\infty$ and use the oriented distance to $0$). This gives us values for the variables of the \Distinct instance. The definition of the addition, multiplication and strict inequality constructions ensures that each of the constraints will be satisfied. Furthermore, by definition of the matroid, all of the variables are distinct. This finishes the proof.

\acknowledgements We are grateful to the anonymous reviewers for their numerous suggestions that greatly improved the article.

\bibliographystyle{abbrvnat}
\bibliography{Finalbiblio}

\end{document}

%% file: FinalFigures/projection.pdf_tex
\begingroup%
  \makeatletter%
  \providecommand\color[2][]{%
    \errmessage{(Inkscape) Color is used for the text in Inkscape, but the package 'color.sty' is not loaded}%
    \renewcommand\color[2][]{}%
  }%
  \providecommand\transparent[1]{%
    \errmessage{(Inkscape) Transparency is used (non-zero) for the text in Inkscape, but the package 'transparent.sty' is not loaded}%
    \renewcommand\transparent[1]{}%
  }%
  \providecommand\rotatebox[2]{#2}%
  \ifx\svgwidth\undefined%
    \setlength{\unitlength}{4027.70280542bp}%
    \ifx\svgscale\undefined%
      \relax%
    \else%
      \setlength{\unitlength}{\unitlength * \real{\svgscale}}%
    \fi%
  \else%
    \setlength{\unitlength}{\svgwidth}%
  \fi%
  \global\let\svgwidth\undefined%
  \global\let\svgscale\undefined%
  \makeatother%
  \begin{picture}(1,0.36767744)%
    \put(0,0){\includegraphics[width=\unitlength,page=1]{projection.pdf}}%
    \put(0.23121697,0.36201839){\color[rgb]{0,0,0}\makebox(0,0)[b]{\smash{$\ell$}}}%
    \put(0.03150642,0.13554496){\color[rgb]{0,0,0}\makebox(0,0)[b]{\smash{$a$}}}%
    \put(0.21147533,0.14028758){\color[rgb]{0,0,0}\makebox(0,0)[b]{\smash{$b$}}}%
    \put(0.14483129,0.18320455){\color[rgb]{0,0,0}\makebox(0,0)[b]{\smash{$c$}}}%
    \put(0.09587634,0.24383757){\color[rgb]{0,0,0}\makebox(0,0)[b]{\smash{$d$}}}%
    \put(0.17296736,0.25642373){\color[rgb]{0,0,0}\makebox(0,0)[b]{\smash{$e$}}}%
    \put(0.39846066,0.14231657){\color[rgb]{0,0,0}\makebox(0,0)[b]{\smash{$f$}}}%
    \put(0.59333316,0.14457379){\color[rgb]{0,0,0}\makebox(0,0)[b]{\smash{$f$}}}%
    \put(0.7395181,0.09265795){\color[rgb]{0,0,0}\makebox(0,0)[b]{\smash{$a$}}}%
    \put(0.89835708,0.13479255){\color[rgb]{0,0,0}\makebox(0,0)[b]{\smash{$b$}}}%
    \put(0.83125933,0.18546176){\color[rgb]{0,0,0}\makebox(0,0)[b]{\smash{$c$}}}%
    \put(0.78850614,0.24690066){\color[rgb]{0,0,0}\makebox(0,0)[b]{\smash{$d$}}}%
    \put(0.91555185,0.25255476){\color[rgb]{0,0,0}\makebox(0,0)[b]{\smash{$e$}}}%
    \put(0,0){\includegraphics[width=\unitlength,page=2]{projection.pdf}}%
  \end{picture}%
\endgroup%

%% file: FinalFigures/additionsimple.pdf_tex
\begingroup%
  \makeatletter%
  \providecommand\color[2][]{%
    \errmessage{(Inkscape) Color is used for the text in Inkscape, but the package 'color.sty' is not loaded}%
    \renewcommand\color[2][]{}%
  }%
  \providecommand\transparent[1]{%
    \errmessage{(Inkscape) Transparency is used (non-zero) for the text in Inkscape, but the package 'transparent.sty' is not loaded}%
    \renewcommand\transparent[1]{}%
  }%
  \providecommand\rotatebox[2]{#2}%
  \ifx\svgwidth\undefined%
    \setlength{\unitlength}{2110.71421417bp}%
    \ifx\svgscale\undefined%
      \relax%
    \else%
      \setlength{\unitlength}{\unitlength * \real{\svgscale}}%
    \fi%
  \else%
    \setlength{\unitlength}{\svgwidth}%
  \fi%
  \global\let\svgwidth\undefined%
  \global\let\svgscale\undefined%
  \makeatother%
  \begin{picture}(1,0.31943604)%
    \put(0,0){\includegraphics[width=\unitlength,page=1]{additionsimple.pdf}}%
    \put(0.29407962,0.00252272){\color[rgb]{0,0,0}\makebox(0,0)[b]{\smash{$\0$}}}%
    \put(0.42269842,0.00252272){\color[rgb]{0,0,0}\makebox(0,0)[b]{\smash{$\x$}}}%
    \put(0.52946003,0.00295646){\color[rgb]{0,0,0}\makebox(0,0)[b]{\smash{$\y$}}}%
    \put(0,0){\includegraphics[width=\unitlength,page=2]{additionsimple.pdf}}%
    \put(0.65669708,0.00295646){\color[rgb]{0,0,0}\makebox(0,0)[b]{\smash{$\x+\y$}}}%
    \put(0,0){\includegraphics[width=\unitlength,page=3]{additionsimple.pdf}}%
    \put(0.95608837,0.08138564){\color[rgb]{0,0,0}\makebox(0,0)[lb]{\smash{$\ell$}}}%
    \put(-3.5833467,1.68275747){\color[rgb]{0,0,0}\makebox(0,0)[lt]{\begin{minipage}{5.16471509\unitlength}\raggedright \end{minipage}}}%
    \put(0,0){\includegraphics[width=\unitlength,page=4]{additionsimple.pdf}}%
  \end{picture}%
\endgroup%

%% file: FinalFigures/addition.pdf_tex
\begingroup%
  \makeatletter%
  \providecommand\color[2][]{%
    \errmessage{(Inkscape) Color is used for the text in Inkscape, but the package 'color.sty' is not loaded}%
    \renewcommand\color[2][]{}%
  }%
  \providecommand\transparent[1]{%
    \errmessage{(Inkscape) Transparency is used (non-zero) for the text in Inkscape, but the package 'transparent.sty' is not loaded}%
    \renewcommand\transparent[1]{}%
  }%
  \providecommand\rotatebox[2]{#2}%
  \ifx\svgwidth\undefined%
    \setlength{\unitlength}{4436.91665752bp}%
    \ifx\svgscale\undefined%
      \relax%
    \else%
      \setlength{\unitlength}{\unitlength * \real{\svgscale}}%
    \fi%
  \else%
    \setlength{\unitlength}{\svgwidth}%
  \fi%
  \global\let\svgwidth\undefined%
  \global\let\svgscale\undefined%
  \makeatother%
  \begin{picture}(1,0.21870914)%
    \put(0,0){\includegraphics[width=\unitlength,page=1]{addition.pdf}}%
    \put(0.05983041,0.00140643){\color[rgb]{0,0,0}\makebox(0,0)[b]{\smash{$\II$}}}%
    \put(0.13372779,0.0012001){\color[rgb]{0,0,0}\makebox(0,0)[b]{\smash{$\0$}}}%
    \put(0.16243472,0.0012001){\color[rgb]{0,0,0}\makebox(0,0)[b]{\smash{$\x$}}}%
    \put(0.2440544,0.00140643){\color[rgb]{0,0,0}\makebox(0,0)[b]{\smash{$\y$}}}%
    \put(0.43941636,0.00140643){\color[rgb]{0,0,0}\makebox(0,0)[b]{\smash{$\x+\y$}}}%
    \put(0.18692331,0.12239021){\color[rgb]{0,0,0}\makebox(0,0)[b]{\smash{$\a$}}}%
    \put(0.26694234,0.17022431){\color[rgb]{0,0,0}\makebox(0,0)[b]{\smash{$\b$}}}%
    \put(0.33410524,0.21357202){\color[rgb]{0,0,0}\makebox(0,0)[lb]{\smash{$\ell_\infty$}}}%
    \put(0.46112739,0.03871649){\color[rgb]{0,0,0}\makebox(0,0)[lb]{\smash{$\ell$}}}%
    \put(0,0){\includegraphics[width=\unitlength,page=2]{addition.pdf}}%
    \put(0.67182088,0.0012001){\color[rgb]{0,0,0}\makebox(0,0)[b]{\smash{$\0$}}}%
    \put(0.75582686,0.0012001){\color[rgb]{0,0,0}\makebox(0,0)[b]{\smash{$\x$}}}%
    \put(0.81136665,0.00140643){\color[rgb]{0,0,0}\makebox(0,0)[b]{\smash{$\y$}}}%
    \put(0,0){\includegraphics[width=\unitlength,page=3]{addition.pdf}}%
    \put(0.89575913,0.00140643){\color[rgb]{0,0,0}\makebox(0,0)[b]{\smash{$\x+\y$}}}%
    \put(0.68464229,0.16259372){\color[rgb]{0,0,0}\makebox(0,0)[b]{\smash{$\c$}}}%
    \put(0.82450516,0.16259372){\color[rgb]{0,0,0}\makebox(0,0)[b]{\smash{$\d$}}}%
    \put(0,0){\includegraphics[width=\unitlength,page=4]{addition.pdf}}%
    \put(0.98175927,0.03871649){\color[rgb]{0,0,0}\makebox(0,0)[lb]{\smash{$\ell$}}}%
    \put(0.18477781,0.05368605){\color[rgb]{0,0,0}\makebox(0,0)[b]{\smash{$\c$}}}%
    \put(0.27127355,0.09787696){\color[rgb]{0,0,0}\makebox(0,0)[b]{\smash{$\d$}}}%
    \put(-1.18227305,0.8005154){\color[rgb]{0,0,0}\makebox(0,0)[lt]{\begin{minipage}{2.45693989\unitlength}\raggedright \end{minipage}}}%
    \put(0,0){\includegraphics[width=\unitlength,page=5]{addition.pdf}}%
  \end{picture}%
\endgroup%

%% file: FinalFigures/multiplication.pdf_tex
\begingroup%
  \makeatletter%
  \providecommand\color[2][]{%
    \errmessage{(Inkscape) Color is used for the text in Inkscape, but the package 'color.sty' is not loaded}%
    \renewcommand\color[2][]{}%
  }%
  \providecommand\transparent[1]{%
    \errmessage{(Inkscape) Transparency is used (non-zero) for the text in Inkscape, but the package 'transparent.sty' is not loaded}%
    \renewcommand\transparent[1]{}%
  }%
  \providecommand\rotatebox[2]{#2}%
  \ifx\svgwidth\undefined%
    \setlength{\unitlength}{4205.08034646bp}%
    \ifx\svgscale\undefined%
      \relax%
    \else%
      \setlength{\unitlength}{\unitlength * \real{\svgscale}}%
    \fi%
  \else%
    \setlength{\unitlength}{\svgwidth}%
  \fi%
  \global\let\svgwidth\undefined%
  \global\let\svgscale\undefined%
  \makeatother%
  \begin{picture}(1,0.23076711)%
    \put(0,0){\includegraphics[width=\unitlength,page=1]{multiplication.pdf}}%
    \put(0.06312901,0.00192939){\color[rgb]{0,0,0}\makebox(0,0)[b]{\smash{$\II$}}}%
    \put(0.09553305,0.00171166){\color[rgb]{0,0,0}\makebox(0,0)[b]{\smash{$\0$}}}%
    \put(0.16346282,0.00171166){\color[rgb]{0,0,0}\makebox(0,0)[b]{\smash{$\x$}}}%
    \put(0.23156575,0.00192939){\color[rgb]{0,0,0}\makebox(0,0)[b]{\smash{$\y$}}}%
    \put(0.46364246,0.00192939){\color[rgb]{0,0,0}\makebox(0,0)[b]{\smash{$\x\y$}}}%
    \put(0.19722885,0.12913788){\color[rgb]{0,0,0}\makebox(0,0)[b]{\smash{$\a$}}}%
    \put(0.28165952,0.17960919){\color[rgb]{0,0,0}\makebox(0,0)[b]{\smash{$\b$}}}%
    \put(0.35252528,0.22534676){\color[rgb]{0,0,0}\makebox(0,0)[lb]{\smash{$\ell_\infty$}}}%
    \put(0.48247376,0.04085102){\color[rgb]{0,0,0}\makebox(0,0)[lb]{\smash{$\ell$}}}%
    \put(0,0){\includegraphics[width=\unitlength,page=2]{multiplication.pdf}}%
    \put(0.60358666,0.00126626){\color[rgb]{0,0,0}\makebox(0,0)[b]{\smash{$\0$}}}%
    \put(0.75668695,0.00126626){\color[rgb]{0,0,0}\makebox(0,0)[b]{\smash{$\x$}}}%
    \put(0.83057642,0.00148398){\color[rgb]{0,0,0}\makebox(0,0)[b]{\smash{$\y$}}}%
    \put(0,0){\includegraphics[width=\unitlength,page=3]{multiplication.pdf}}%
    \put(0.92178748,0.00148398){\color[rgb]{0,0,0}\makebox(0,0)[b]{\smash{$\x\y$}}}%
    \put(0.71826554,0.093726){\color[rgb]{0,0,0}\makebox(0,0)[b]{\smash{$\c$}}}%
    \put(0.84298038,0.13192128){\color[rgb]{0,0,0}\makebox(0,0)[b]{\smash{$\d$}}}%
    \put(0,0){\includegraphics[width=\unitlength,page=4]{multiplication.pdf}}%
    \put(0.13761752,0.00171166){\color[rgb]{0,0,0}\makebox(0,0)[b]{\smash{$\1$}}}%
    \put(0,0){\includegraphics[width=\unitlength,page=5]{multiplication.pdf}}%
    \put(0.97987621,0.04085102){\color[rgb]{0,0,0}\makebox(0,0)[lb]{\smash{$\ell$}}}%
    \put(0.19305058,0.05868423){\color[rgb]{0,0,0}\makebox(0,0)[b]{\smash{$\c$}}}%
    \put(0.27947397,0.10836903){\color[rgb]{0,0,0}\makebox(0,0)[b]{\smash{$\d$}}}%
    \put(0.71518637,0.00126626){\color[rgb]{0,0,0}\makebox(0,0)[b]{\smash{$\1$}}}%
    \put(0.2796327,0.17869212){\color[rgb]{0,0,0}\makebox(0,0)[lb]{\smash{}}}%
    \put(0,0){\includegraphics[width=\unitlength,page=6]{multiplication.pdf}}%
  \end{picture}%
\endgroup%

%% file: FinalFigures/triples.pdf_tex
\begingroup%
  \makeatletter%
  \providecommand\color[2][]{%
    \errmessage{(Inkscape) Color is used for the text in Inkscape, but the package 'color.sty' is not loaded}%
    \renewcommand\color[2][]{}%
  }%
  \providecommand\transparent[1]{%
    \errmessage{(Inkscape) Transparency is used (non-zero) for the text in Inkscape, but the package 'transparent.sty' is not loaded}%
    \renewcommand\transparent[1]{}%
  }%
  \providecommand\rotatebox[2]{#2}%
  \ifx\svgwidth\undefined%
    \setlength{\unitlength}{1144.28567622bp}%
    \ifx\svgscale\undefined%
      \relax%
    \else%
      \setlength{\unitlength}{\unitlength * \real{\svgscale}}%
    \fi%
  \else%
    \setlength{\unitlength}{\svgwidth}%
  \fi%
  \global\let\svgwidth\undefined%
  \global\let\svgscale\undefined%
  \makeatother%
  \begin{picture}(1,0.5940896)%
    \put(0,0){\includegraphics[width=\unitlength,page=1]{triples.pdf}}%
    \put(0.05530837,0.00385322){\color[rgb]{0,0,0}\makebox(0,0)[b]{\smash{$a$}}}%
    \put(0.60570381,0.00385322){\color[rgb]{0,0,0}\makebox(0,0)[b]{\smash{$d$}}}%
    \put(0.875073,0.00385322){\color[rgb]{0,0,0}\makebox(0,0)[b]{\smash{$b$}}}%
    \put(0.58912688,0.34296518){\color[rgb]{0,0,0}\makebox(0,0)[b]{\smash{$c$}}}%
    \put(0.39610119,0.24371424){\color[rgb]{0,0,0}\makebox(0,0)[b]{\smash{$c'$}}}%
    \put(0.21826169,0.45626106){\color[rgb]{0,0,0}\makebox(0,0)[b]{\smash{$e$}}}%
  \end{picture}%
\endgroup%

%% file: FinalFigures/hyperplane.pdf_tex
\begingroup%
  \makeatletter%
  \providecommand\color[2][]{%
    \errmessage{(Inkscape) Color is used for the text in Inkscape, but the package 'color.sty' is not loaded}%
    \renewcommand\color[2][]{}%
  }%
  \providecommand\transparent[1]{%
    \errmessage{(Inkscape) Transparency is used (non-zero) for the text in Inkscape, but the package 'transparent.sty' is not loaded}%
    \renewcommand\transparent[1]{}%
  }%
  \providecommand\rotatebox[2]{#2}%
  \ifx\svgwidth\undefined%
    \setlength{\unitlength}{1667.275716bp}%
    \ifx\svgscale\undefined%
      \relax%
    \else%
      \setlength{\unitlength}{\unitlength * \real{\svgscale}}%
    \fi%
  \else%
    \setlength{\unitlength}{\svgwidth}%
  \fi%
  \global\let\svgwidth\undefined%
  \global\let\svgscale\undefined%
  \makeatother%
  \begin{picture}(1,0.59910042)%
    \put(0,0){\includegraphics[width=\unitlength,page=1]{hyperplane.pdf}}%
    \put(0.89888994,0.36194106){\color[rgb]{0,0,0}\makebox(0,0)[lb]{\smash{$P$}}}%
    \put(0.7271258,0.27106054){\color[rgb]{0,0,0}\makebox(0,0)[lb]{\smash{$d$}}}%
    \put(0.71894655,0.44645994){\color[rgb]{0,0,0}\makebox(0,0)[lb]{\smash{$a$}}}%
    \put(0.62170437,0.40192849){\color[rgb]{0,0,0}\makebox(0,0)[lb]{\smash{$a'$}}}%
    \put(0.53536789,0.52916122){\color[rgb]{0,0,0}\makebox(0,0)[lb]{\smash{$x$}}}%
    \put(0.34451881,0.2283467){\color[rgb]{0,0,0}\makebox(0,0)[lb]{\smash{$a_1$}}}%
    \put(0,0){\includegraphics[width=\unitlength,page=2]{hyperplane.pdf}}%
  \end{picture}%
\endgroup%